\documentclass{ita}
\usepackage{amssymb}
\usepackage{amsmath}
\usepackage{multirow}
\usepackage{wasysym}
\usepackage{textcomp}
\usepackage{enumerate}
\usepackage{comment}
\usepackage{graphicx}
\usepackage{enumerate,tikz}

\newtheorem{fact}[thrm]{Fact}
\begin{document}

\title{LANGUAGE CLASSES ASSOCIATED WITH AUTOMATA OVER MATRIX GROUPS}

\thanks{This work was supported by Bo\u{g}azi\c{c}i University Research Fund under grant number 11760. A preliminary version of this work was presented at the 8'th Workshop on Non-Classical Models of Automata and Applications (NCMA), Debrecen, Hungary, August 29-30, 2016.} 
\thanks{\"{O}zlem Salehi was partially
supported by T\"{U}B\.{I}TAK (Scientific and Technological Research Council of Turkey).}
\thanks{The research of F. D'Alessandro was supported by a EC-FP7 Marie Curie-T\"UB\.{I}TAK Co-Funded Brain Circulation Scheme Project 2236 Fellowship.}
\author{\"{O}zlem Saleh\.{i}}
\address{Bo\u{g}azi\c{c}i University, Department of Computer Engineering, Bebek 34342 \.{I}stanbul, Turkey \email{ozlem.salehi@boun.edu.tr \& say@boun.edu.tr}}
\author{Flavio D'Alessandro}
\address{Bo\u gazi\c ci University, 
	Department of Mathematics, Bebek 34342, \.{I}stanbul, Turkey
		}
	\secondaddress{Universit\`a di Roma ``La Sapienza'', Dipartimento di Matematica, Piazzale Aldo Moro 2, 00185 Roma, Italy \email{dalessan@mat.uniroma1.it }}

\author{A. C. Cem Say}
\sameaddress{1}

\date{...}

\begin{abstract} 
We investigate the language classes recognized by group automata over
matrix groups. For the case of $2 \times 2 $ matrices, we prove that
the corresponding group automata for rational matrix groups are more
powerful than the corresponding group automata for integer matrix
groups. Finite automata over some special matrix groups, such as the
discrete Heisenberg group and the Baumslag-Solitar group are also
examined.  We also introduce the notion of time complexity for group
automata and demonstrate some separations among related classes. The
case of linear-time bounds is examined in detail throughout our
repertory of matrix group automata.
\end{abstract}

\maketitle
\section{Introduction}

Many extensions of the classical finite automaton model have been
examined. One such variant is the group automaton (finite automaton
over groups), which is a nondeterministic finite automaton equipped
with a register that holds an element from a group \cite{MS97}. The
register is initialized to the identity element of the group, and a
computation is deemed successful if the register is equal to the
identity element at the end of the computation after being multiplied
at every step. This setup generalizes various models such as
nondeterministic blind multicounter automata \cite{FMR67} and finite
automata with multiplication \cite{ISK76}.

The theory of group automata has been essentially developed in the
case of free groups \cite{DM00,Co05,Ka09}, and in the case of free
Abelian groups \cite{EO04,EKO08}, where strong theorems allow to
characterize the power of such models and the combinatorial properties
of the languages recognized by these automata. For groups that are not of the types mentioned above, even in the case
of groups of matrices of low dimension, the study of group automata quickly beomes nontrivial, and there are remarkable classes of linear
groups for which little is known about the automaton models that they
define.

In this paper, we present several new results about the classes of
languages recognized by finite automata over matrix groups. We focus
on matrix groups with integer and rational entries. For the case of $
2 \times 2 $ matrices, we prove that the corresponding group automata
for rational matrix groups are more powerful than the corresponding
group automata for integer matrix groups. We also explore finite
automata over some special matrix groups, such as the discrete
Heisenberg group and the Baumslag-Solitar group. The ``zoo" of
language classes associated with different groups is presented,
visualizing known relationships and open problems.

We also introduce the notion of time complexity for group automata,
and use this additional dimension to analyze the relationships among
the language families of various automata using different groups. We
develop a method for proving that automata over matrix groups where the growth rate of the group and the time are bounded can not recognize certain languages, even if one uses a very weak definition of time-bounded computation, and use this to demonstrate
some new relationships between time-bounded versions of our language
classes. The case of linear-time bounds is examined in detail
throughout our repertory of matrix groups.

\section{Preliminaries} \label{Section: pre}

\subsection{Notation and terminology}

The following notation will be used throughout the paper: $Q$ is the set of
states, $q_0 \in Q$ denotes the initial state, $Q_a \subseteq Q$ denotes the
set of accepting states, and $\Sigma$ is the input alphabet.

By $ w^r $, we represent the 
reverse of the string $ w $. The length of $ w $ is denoted by $ |w| $. 

$ \mathsf{REG} $, $ \mathsf{CF} $, and $ \mathsf{RE} $ denote the families of 
regular languages, context-free languages, and recursively enumerable languages,
respectively.

We assume a familiarity with some basic notions from algebra and group theory 
(see {\cite{Fr03},\cite{LS77}} for references on this topic). For a finitely generated 
group $ G $ and a set $ X $ of generators, the word problem language of $ G $ is 
the language $ W(G,X) $ over $ \Sigma =\{X \cup X^{-1}\} $ which consists of all words that represent the 
identity element of $ G $. Most of the time, the statements about the word problem 
are independent of the generating set and in these cases the word problem 
language is denoted by $ W(G) $. For a string $ w=w_1w_2\dots w_n \in W(G)  $, $ w^{-1}=w_n^{-1}\dots w_1^{-1} $ where each $ w_i \in \Sigma$ represents a generator.

\subsection{Group automata}

Group automata first appear explicitly in the paper \cite{MS97} under the name of extended finite automaton. 
The definition is formally given as follows.

Let $ K = (M,\circ, e )$ be a group under the operation denoted by 
$ \circ $ with the neutral element denoted by $ e $. 
An \textit{extended finite automaton} over the group $ K = (M,\circ, e)$ is a
6-tuple
\[ \mathcal{F} = (Q, \Sigma,K,\delta, q_0,Q_a), \]
where  the transition function $\delta$ is defined as
\[\delta: Q \times (\Sigma \cup \{\varepsilon\}) \rightarrow \mathbb{P}(Q\times M).\] 
$ \delta(q,\sigma) \ni (q',m) $ means that when $\mathcal{F}$ reads the
symbol (or empty string) $\sigma \in \Sigma \cup \{\varepsilon\}$
in state $q$, it can move to state $q'$, and write $ x\circ m $ in the register, 
where $ x $ is the old content of the register. The initial value of the register is the 
neutral element $ e $ of the group $ K $. The string is accepted if, after
completely reading the string,  $\mathcal{F}$ enters an accept state with the 
content of the register being equal to the neutral element of $ K $.

We will prefer using the name group automaton ($ G $-automaton) instead of 
extended finite automaton over group $ G $.

\textit{Monoid automata} are defined analogously where the group $ G $ is replaced by some monoid $ N $.

The class of languages recognized by $ G $-automata will be denoted as 
$ \mathfrak{L}(G) $.

\section{Matrix groups and associated language classes}
\label{Section: results}

In this section, we are going to prove some new results about the
classes of languages recognized by finite automata over various groups,
focusing on linear groups.

\subsection{Basic results}

We will denote the \textit{free group} over $r$ generators by $ \mathbf{F}_r $.
Note that $ \mathbf{F}_0 $ is the trivial group, and $ \mathbf{F}_1 $
is isomorphic to $ \mathbb{Z} $, the additive group of integers. The
class of regular languages is characterized as the set of languages
recognized by finite automata over the trivial group $\mathbf{F}_0 $
in  \cite{DM00}.

The relation between the classes of languages recognized by free group
automata is summarized as follows.

\begin{fact}\label{fact: reg} \textup{\cite{DM00}}		
	$ \mathsf{REG} = \mathfrak{L}(\mathbf{F}_0) \subsetneq
	\mathfrak{L}(\mathbf{F}_1) = \mathfrak{L}(\mathbb{Z}) \subsetneq
	\mathfrak{L}(\mathbf{F}_2) $.
\end{fact}

A characterization of context-free languages by group automata was
first stated by Dassow and Mitrana \cite{DM00}, and proven in \cite{Co05}. 
Let us note that $ \mathbf{F}_2 $ contains any free group
of rank $ n \geq 2 $ \cite{LS77}.

\begin{fact}\textup{\cite{DM00, Co05, Ka09}}\label{fact: cf}
	$\mathfrak{L}(\mathbf{F}_2)$ is the family of context-free languages.
\end{fact}

We will denote by $ \mathbb{Z}^k $ the additive group of integer
vectors of dimension $k$. This group is
isomorphic to the \textit{free Abelian group of rank $ k$}, and  $
\mathbb{Z}^k $-automata are equivalent to nondeterministic blind
$k$-counter automata \cite{Gr78}.

The following result states the hierarchy between the classes of languages 
recognized by $ \mathbb{Z}^k $-automata. This result also follows from the
hierarchy between the class of languages recognized by nondeterministic 
blind $ k $-counter automata.

\begin{fact}\textup{\cite{CEO06}}
	$ \mathfrak{L}(\mathbb{Z}^k) \subsetneq
	\mathfrak{L}(\mathbb{Z}^{k+1})$ for $ k \geq 1 $.
\end{fact}

 We denote by $ \mathbb{Q}^+$ the
 multiplicative group of positive rational numbers, which is isomorphic
 to a free Abelian group of infinite rank. A $ \mathbb{Q}^+$-automaton
 is also equivalent to a nondeterministic finite automaton with
 multiplication without equality (1NFAMW) of Ibarra et al.
 \cite{ISK76}.

The following fact characterizes the class of
languages recognized for the case where the alphabet is unary,

\begin{fact}\textup{\cite{ISK76}}\label{fact: unary}
	All \textup{1NFAMW}-recognizable languages over a unary alphabet are regular.
\end{fact}

Let us mention that the class of context-free languages and the class
of languages recognized by nondeterministic blind counter automata are
incomparable.

\begin{fact}\label{fact: cfzn}
	$\mathsf{CF}$ and $\mathfrak{L}(\mathbb{Z}^k)$ are incomparable for
	all $ k \geq 2 $.		
\end{fact}

\begin{proof} Consider the language $\mathtt{L}=\{a^nb^n| n\geq 0\}$ which is a
	context-free language. Since context-free languages are closed under star, 
	$\mathtt{L}^*$ is a context-free language whereas it cannot be recognized 
	by any $\mathbb{Z}^k$-automaton for all $ k \geq 1 $ by
	\cite{Gr78}. On the other hand, the non-context-free language 
	$\mathtt{L}'=\{a^nb^nc^n|n \geq 0\}$ can be recognized by a
	$\mathbb{Z}^2$-automaton.
\end{proof}

\subsection{Automata on groups of $ 2\times2 $ and $ 3\times3 $ matrices}\label{Section: 23matrices}
We denote by $GL(2,\mathbb{Z})$ the general linear group of degree two
over the field of integers, that is, the
group of $ 2\times 2 $ invertible matrices with integer entries. Note
that these matrices have determinant $\pm 1$. Restricting the matrices
in $GL(2,\mathbb{Z})$ to those that have determinant 1, we obtain the
special linear group of degree two over the field of integers,
$SL(2,\mathbb{Z})$.

Let $ \mathbf{G} $ be the group generated by the  matrices
\[
M_{a}=
\left [
\begin{array}{cc}
1&2\\
0&1\\
\end{array}
\right ]~~~\mbox{and}~~~
M_{b}=
\left [
\begin{array}{cc}
1&0\\
2&1\\
\end{array}
\right ].
\]
There exists an isomorphism $ \varphi $ from $\mathbf{F}_2 $ onto
$\mathbf{G} $ by \cite{KM79}. Note that $ M_a $ and $ M_b $ are
integer matrices with determinant 1, which proves that $ \mathbf{F}_2
$ is a subgroup of $ SL(2,\mathbb{Z}) $.

Now the question is whether $ \mathfrak{L}(GL(2,\mathbb{Z}))$ and $
\mathfrak{L}(SL(2,\mathbb{Z}))$ correspond to larger classes of
languages than the class of context-free languages. We are going to
use the following fact to prove that the answer is negative.

\begin{fact}\textup{\cite{Co05}}  \label{fact: finite}
	Suppose $G$ is a finitely generated group and $H$ is a subgroup of
	finite index. Then $\mathfrak{L}(G) = \mathfrak{L}(H)$.
\end{fact}

Now we are ready to state our theorem.

\begin{thrm}\label{theorem:gl}
	$ \mathsf{CF} =\mathfrak{L}(\mathbf{F}_2) = \mathfrak{L}(SL(2,\mathbb{Z})) = \mathfrak{L}(GL(2,\mathbb{Z}))$. 
\end{thrm}
\begin{proof} We are going to use Fact  \ref{fact: finite} to prove the result. 
	Since $ SL(2,\mathbb{Z}) $ has index 2 in $ GL(2,\mathbb{Z})$
	and $GL(2,\mathbb{Z})$ is finitely generated, $\mathfrak{L}(GL(2,\mathbb{Z})) =
	\mathfrak{L}(SL(2,\mathbb{Z}))$. Since $\mathbf{F}_2$ has index 12 in 
	$SL(2,\mathbb{Z})$ \cite{BO08} and $SL(2,\mathbb{Z})$ is finitely generated, $
	\mathfrak{L}(SL(2,\mathbb{Z})) = \mathfrak{L}(\mathbf{F}_2)$ which is equal to 
	the family of context-free languages by Fact \ref{fact: cf}.
\end{proof}

Let us now investigate the group $SL(3,\mathbb{Z})$, the group of 
$3 \times 3 $ integer matrices with determinant~$1$.

We start by looking at an important subgroup of $SL(3,\mathbb{Z})$,
the discrete Heisenberg group. The discrete Heisenberg group
$\mathbf{H}$ is defined as $\langle a,b | ab=bac,ac=ca,bc=cb
\rangle$,
where
$c=a^{-1}b^{-1}ab$ is the commutator of $a$ and $b$.

$$ a=\left [
\begin{array}{ccc}
1&1&0\\
0&1&0\\
0&0&1
\end{array}
\right ]~~~
b=\left [
\begin{array}{ccc}
1&0&0\\
0&1&1\\
0&0&1
\end{array}
\right ]~~~
c=\left [
\begin{array}{ccc}
1&0&1\\
0&1&0\\
0&0&1
\end{array}
\right ]
$$

Any element $g \in \mathbf{H}$ can be written uniquely as $b^ja^ic^k$.
$$ g=\left [
\begin{array}{ccc}
1&i&k\\
0&1&j\\
0&0&1
\end{array}
\right ] = b^ja^ic^k
$$

It is shown in \cite{Re10} that the languages
$\mathtt{MULT}=\{x^py^qz^{pq} | p,q\geq 0\}$,
$\mathtt{COMPOSITE}=\{x^{pq}| p,q >1\}$ and
$\mathtt{MULTIPLE}=\{x^py^{pn}|p \in \mathbb{N}\}$ can be recognized
by $\mathbf{H}$-automata, using the special multiplication property
of the group.

Correcting a small error in \cite{Re10}, we rewrite the multiplication
property of the elements
of $\mathbf{H}$.

\begin{equation*}\label{eq0}
(b^xa^yc^z)(b^{x'}a^{y'}c^{z'})= b^{x+x'}a^{y+y'}c^{z+z'+yx'} 
\end{equation*}

We can make the following observation using the fact that
$\mathfrak{L}(\mathbf{H})$ contains non-context-free languages.

\begin{thrm}\label{thm: sl2z3z}
	$\mathfrak{L}(SL(2,\mathbb{Z})) \subsetneq \mathfrak{L}(SL(3,\mathbb{Z}))$.
\end{thrm}
\begin{proof}
	It is obvious that an $SL(2,\mathbb{Z})$-automaton can be simulated by
	an $SL(3,\mathbb{Z})$-automaton.
	Note that $\mathfrak{L}(SL(2,\mathbb{Z}))$ is the family of
	context-free languages by Theorem \ref{theorem:gl}. Since
	$\mathfrak{L}(\mathbf{H}) \subsetneq \mathfrak{L}(SL(3,\mathbb{Z}))$
	and the non-context-free language $\mathtt{MULT}=\{x^py^qz^{pq} | p,q\geq 0\}$ can be recognized by an $ \mathbf{H} $-automaton \cite{Re10}, the
	result follows.
\end{proof}

The following result
is a direct consequence of Fact \ref{fact: finite}.

\begin{thrm}\label{thm: sl3zgl3z}
	$ \mathfrak{L}(SL(3,\mathbb{Z}))=\mathfrak{L}(GL(3,\mathbb{Z})) $.
\end{thrm}

\begin{proof} Since $ GL(3,\mathbb{Z}) $ is a finitely generated group and $
	SL(3,\mathbb{Z}) $ has finite index in $ GL(3,\mathbb{Z}) $, the
	result follows by Fact \ref{fact: finite}.
\end{proof}

We have talked about the discrete Heisenberg group $ \textbf{H} $. Now let us look at a
 subgroup of $\mathbf{H}$ generated by the matrices $B$ and $C$, which
 we will call $\mathbf{G_2}$.
 
 $$ B=\left [	
 \begin{array}{ccc}
 1&0&0\\
 0&1&1\\
 0&0&1
 \end{array}
 \right ]~~~
 C=\left [
 \begin{array}{ccc}
 1&0&1\\
 0&1&0\\
 0&0&1
 \end{array}
 \right ]~~~
 $$	
 $\mathbf{G_2}=\langle B,C |BC=CB\rangle $ is a free Abelian group of
 rank 2, and therefore it is isomorphic to $\mathbb{Z}^2$.
 
 We conclude the following about the language recognition power of $
 \mathbb{Z}^2$ and $\mathbf{H}$.
 
 \begin{thrm}\label{theorem: Z2H}
 	$\mathfrak{L}(\mathbb{Z}^2) \subsetneq \mathfrak{L}(\mathbf{H})$.
 \end{thrm}
 \begin{proof}
 	Since $\mathbb{Z}^2$ is a subgroup of $\mathbf{H}$,
 	$\mathfrak{L}(\mathbb{Z}^2) \subseteq 
 	\mathfrak{L}(\mathbf{H}) $ follows. The inclusion is proper since
 	$\mathbf{H}$ can recognize the language 
    $\mathtt{MULT}=\{x^py^qz^{pq} | p,q\geq 0\}$
 	\cite{Re10}, whereas any bounded language in $ \mathfrak{L}(\mathbb{Q}^+) $ is semilinear \cite{ISK76}.
 \end{proof}

Now let us move on to the discussion about matrix groups with rational entries.

Let us denote by $GL(2,\mathbb{Q})$ the general linear group of degree
two over the field of rational  numbers, that is, the group of
invertible matrices with rational entries.  Restricting the matrices
in $GL(2,\mathbb{Q})$ to those that have determinant 1, we obtain the
special  linear group of degree two over the field of rationals,
$SL(2,\mathbb{Q})$.

We will start by proving that allowing rational entries enlarges the
class of languages recognized by matrices with determinant 1.

\begin{thrm}\label{theorem: UPOWODD}
	$ \mathfrak{L}(SL(2,\mathbb{Z})) \subsetneq \mathfrak{L}(SL(2,\mathbb{Q}))  $.
\end{thrm}
\begin{proof} It is obvious that $ \mathfrak{L}(SL(2,\mathbb{Z}))
	\subseteq \mathfrak{L}(SL(2,\mathbb{Q}))  $. We will prove that the
	inclusion is proper.
	
	Let us construct an $SL(2,\mathbb{Q})$-automaton $\mathcal{G}$
	recognizing the language $\mathtt{L}=\{a^{2^{2n+1}}|$  $ n \geq 0 \}$. 
	The state diagram of $ \mathcal{G} $ and the matrices are given in 
	Figure \ref{fig:uoddpow}. Without scanning
	any input symbol, $\mathcal{G}$ first multiplies its register with
	the matrix $A_1$. $\mathcal{G}$ then multiplies its register with the
	matrix $ A_2 $ successively until nondeterministically moving to the
	next state. After that point, $\mathcal{G}$ starts reading the string
	and multiplies its register with the matrix $A_3$ for each scanned
	$a$. At some point, $\mathcal{G}$ nondeterministically stops reading
	the rest of the string and multiplies its register with the matrix
	$A_4$. After successive multiplications with $A_4$, $\mathcal{G}$
	nondeterministically decides moving to an accept state. 
	
	\begin{figure}[h]
		\includegraphics[width=0.8\linewidth]{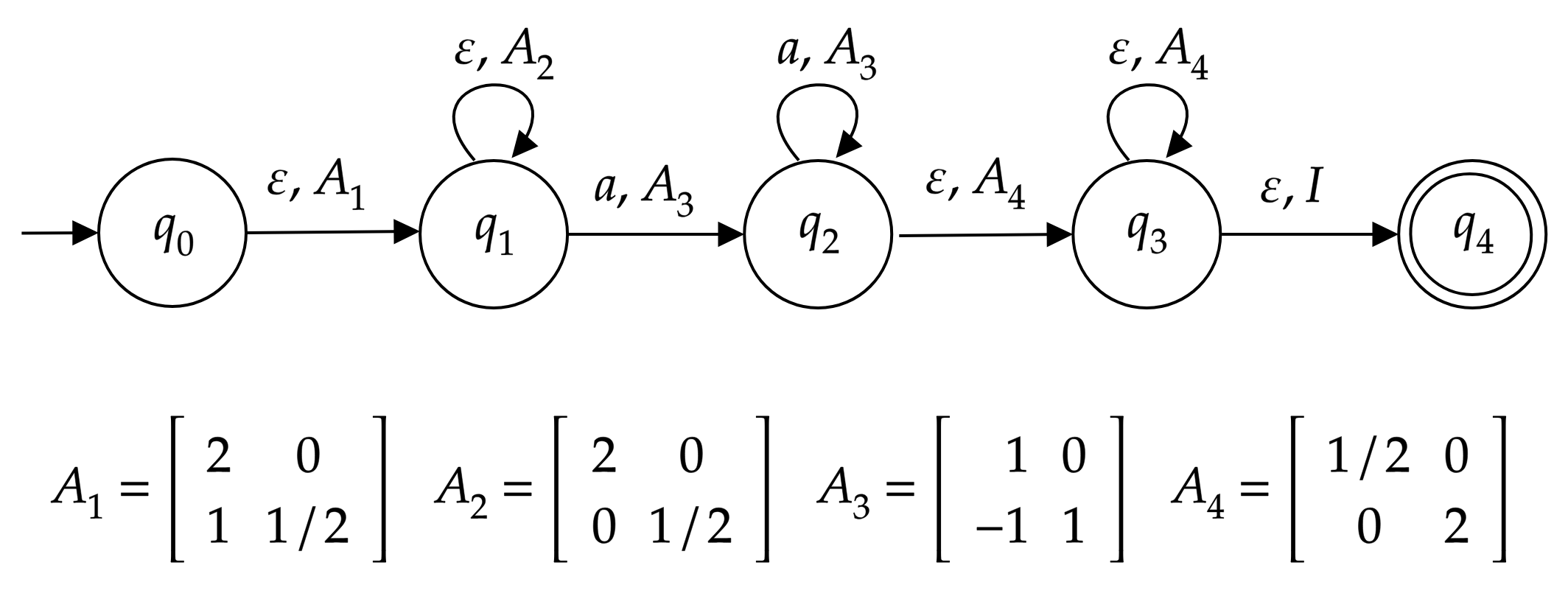}

		\caption{State diagram of $ \mathcal{G} $ accepting the language 
			$\mathtt{L}=\{a^{2^{2n+1}} | n \geq 0 \}$ }
		\label{fig:uoddpow}	
	\end{figure}

	Let us trace the value of the register at different stages of the
	computation. Before reading the first input symbol, the register has
	the value
	$$\left [
	\begin{array}{cc}
	2^{x+1}&0\\
	2^x  &\frac{1}{2^{x+1}}\\
	\end{array}
	\right ]$$
	
	\noindent as a result of the multiplications with the matrix $A_1$
	and $ x $ times the matrix $ A_2 $. Multiplication
	with each $A_3$ leaves $2^{x+1}$ and $ \frac{1}{2^{x+1}} $ unchanged
	while subtracting $ \frac{1}{2^{x+1}} $ from $2^x$ for each scanned
	$a$. As a result of $ y $  multiplications with $ A_3 $, the register
	will have the value
	
	$$\left [
	\begin{array}{cc}
	2^{x+1}&0\\
	2^x-\frac{y}{2^{x+1}} &\frac{1}{2^{x+1}}\\
	\end{array}
	\right ].$$
	
	For the rest of the computation, $\mathcal{G}$ will multiply its
	register with $A_4$ until
	nondeterministically moving to the final state. As a result of $ z $
	multiplications with $ A_4 $, the register will have the value
	
	$$\left [
	\begin{array}{cc}
	\frac{2^{x+1}}{2^z}&0\\
	\bigl( 2^x-\frac{y}{2^{x+1}}\bigr)\frac{1}{2^z} &\frac{2^z}{2^{x+1}}\\
	\end{array}
	\right ].$$
	
	The final value of the register is equal to the identity matrix when
	$ y=2^{2x+1} $ and $ z=x+1 $, which is possible only when the length
	of the input string is $ 2^{2x+1} $ for some $ x\geq 0 $. In the
	successful branch, the register will be equal to the identity matrix
	and $\mathcal{G}$ will end up in the final state having successfully
	read the input string. For input strings which are not members of
	$\mathtt{L}$, either the computation will end before reading the whole
	input string, or the final state will be reached with the register
	value not equaling the identity matrix.
	
	Since the matrices used during the computation are 2 by 2 rational
	matrices with determinant 1, $ \mathtt{L} \in
	\mathfrak{L}(SL(2,\mathbb{Q})) $. $ \mathfrak{L}(SL(2,\mathbb{Q})) $
	contains a unary nonregular language, which is not true for $
	\mathfrak{L}(SL(2,\mathbb{Z})) $ by Theorem \ref{theorem:gl} and we
	conclude the result.
\end{proof}

Let us note that the set of languages recognized by $
\mathbb{Q}^+ $-automata is a proper subset of the set of languages
recognized by $ SL(2,\mathbb{Q}) $-automata, which can be concluded 
with the help of the following fact.

\begin{thrm}\label{thm: qsl2q}
	$ \mathfrak{L}(\mathbb{Q}^+) \subsetneq \mathfrak{L}(SL(2,\mathbb{Q}))  $.
\end{thrm}

\begin{proof} Let $ \mathtt{L} \in \mathfrak{L}(\mathbb{Q}^+) $ and let $
	\mathcal{G} $ be a $ \mathbb{Q}^+ $-automaton recognizing $ \mathtt{L}
	$. We will construct an $ SL(2,\mathbb{Q}) $-automaton  $
	\mathcal{G}' $ recognizing $ \mathtt{L} $. Let $ S =\{s_1,\dots,s_n\}
	$ be the set of elements multiplied with the register during the
	computation of $ \mathcal{G} $. We define the mapping $ \varphi $ as
	follows.
	\[ \varphi: s_i
	\mapsto
	\left [
	\begin{array}{cc}
	s_i&0\\
	0&\frac{1}{s_i}\\
	\end{array}
	\right ]~~~
	\]
	The elements $\varphi(s_i) $ are $ 2 \times 2 $ rational matrices with
	determinant 1. Let $ \delta $ and $ \delta' $ be the transition
	functions of $ \mathcal{G} $ and $ \mathcal{G}' $ respectively. We let
	$ (q',s_i) \in \delta(q,\sigma) \iff (q',\varphi(s_i)) \in
	\delta'(q,\sigma)$ for every $ q,q' \in Q$, $ \sigma \in \Sigma $ and
	$ s_i \in S $. The resulting $ \mathcal{G}' $ recognizes $ \mathtt{L}$.
	
	The inclusion is proper since  $\mathtt{L}=\{a^{2^{2n+1}} | n \geq 0
	\} \in \mathfrak{L}(SL(2,\mathbb{Q}))$ by Theorem \ref{theorem:
		UPOWODD}, and $ \mathfrak{L}(\mathbb{Q}^+) $ does not contain any
	unary nonregular languages by Fact \ref{fact: unary}, noting that 
	$ \mathbb{Q}^+ $-automata are equivalent to 1NFAMW's.
\end{proof}

We will now look at a special subgroup of $GL(2,\mathbb{Q})$.

For two integers $m$ and $n$, the \textit{Baumslag-Solitar group}
$BS(m,n)$ is defined as $BS(m,n)=\langle
a,b | ba^mb^{-1}=a^n\rangle$. We are going to focus on
$BS(1,2)=\langle a,b|bab^{-1}=a^2\rangle$.

Consider the matrix group $G_{BS}$ generated by the matrices

$$A=\left [
\begin{array}{cc}
1&0\\
-1 &1\\
\end{array}
\right ]~~~\mbox{and}~~~
B=\left [
\begin{array}{cc}
1/2&0\\
0 &1\\
\end{array}
\right ].$$

Consider the isomorphism $a \mapsto A$, $b \mapsto B$. The matrices
$A$ and $B$ satisfy the property $BAB^{-1} = A^2$,
\[ \left [
\begin{array}{cc}
1/2&0\\
0 &1\\
\end{array}
\right ]
\left [
\begin{array}{cc}
1&0\\
-1 &1\\
\end{array}
\right ]
\left [
\begin{array}{cc}
2&0\\
0&1\\
\end{array}
\right ]=
\left [
\begin{array}{cc}
1&0\\
-2&1\\
\end{array}
\right ],  \]
and we conclude that $G_{BS}  $ is isomorphic to $BS(1,2)$.

We will prove that there exists a $ BS(1,2) $-automaton which
recognizes a non-context-free language.

\begin{thrm}\label{theorem: upow}
	$ \mathfrak{L}(BS(1,2)) \nsubseteq \mathsf{CF} $.
\end{thrm}
\begin{proof}
	Let us construct a $BS(1,2)$-automaton $\mathcal{G}$ recognizing the
	language $\mathtt{UPOW}=\{a^{2^n}|n\geq 0\}$. The state diagram of 
	$ \mathcal{G} $ and the matrices are given in Figure \ref{fig:upow}. 
	Without scanning any input symbol, $\mathcal{G}$ multiplies its register 
	with the matrix $A_1$ successively. $\mathcal{G}$
	nondeterministically moves to the next state reading the first input
	symbol without modifying the register.
	After that point, $\mathcal{G}$ starts reading the string and
	multiplies its register with the matrix $A_2$
	for each scanned $a$. At some point, $\mathcal{G}$
	nondeterministically stops reading
	the rest of the string and multiplies its register with the element
	$A_3$. After successive multiplications
	with $A_3$, $\mathcal{G}$ nondeterministically decides to move to an
	accept state. 
	\begin{figure}[!htb]
		\centering
		\includegraphics[width=0.6\linewidth]{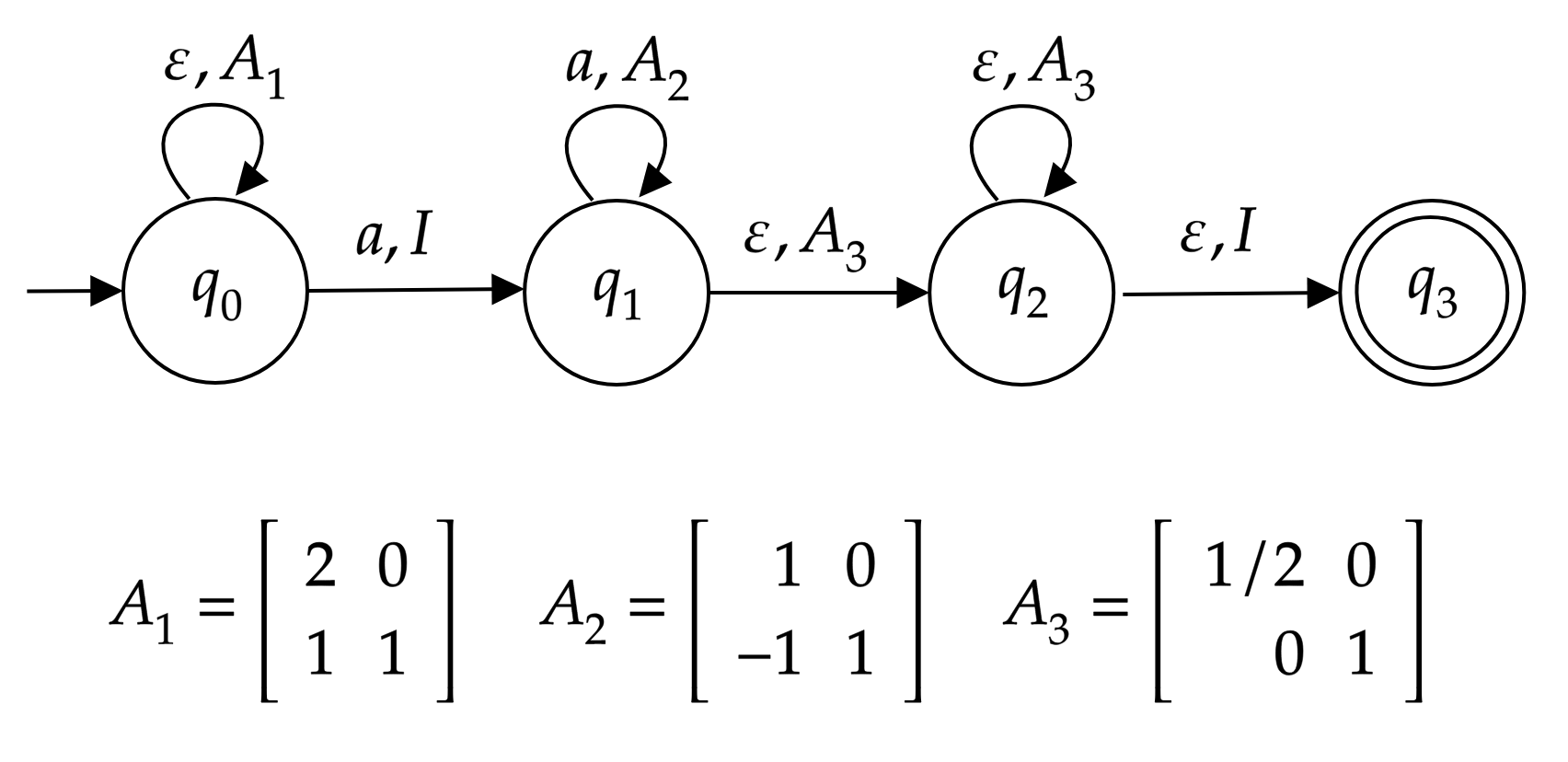}
		\caption{State diagram of $ \mathcal{G} $ recognizing
			\label{fig:upow}
			$\mathtt{UPOW}=\{a^{2^n}|n\geq 0\}$}
	\end{figure}

	As a result of $ i $ multiplications with $ A_1 $, the register has the value 
	
	$$\left [
	\begin{array}{cc}
	2^i&0\\
	2^i -1 &1\\
	\end{array}
	\right ]$$
	before reading the first input symbol. Multiplication
	with each $A_2$ leaves $2^i$ unchanged while subtracting 1 from $2^i -
	1$ for each scanned $a$. The register will have the value 
	$$\left [
	\begin{array}{cc}
	2^i&0\\
	2^i -1 - j &1\\
	\end{array}
	\right ]$$
	as a result of $ j $ multiplications with the matrix $ A_2 $.
	
	For the rest of the computation, $\mathcal{G}$ will multiply its
	register with $A_3$ resulting in the register value 
	
	$$\left [
	\begin{array}{cc}
	\frac{2^i}{2^k}&0\\
	2^i -1 - j &1\\
	\end{array}
	\right ]$$
since each multiplication with $ A_3 $ divides $ 2^i $ by 2.
	
	The register contains the identity matrix at the end of the computation if $ i=k $ and $ j=2^i-1 $ which is possible if the input string is of the form $ a^{1+2^{i}-1} = a^{2^i} $. In the successful
	branch, the register will be equal to the
	identity matrix and $\mathcal{G}$ will end up in the final state
	having successfully read the input string.
	
	For input strings which are not members of $\mathtt{UPOW}$, either the
	computation will end before reading the
	whole input string or the final state will be reached with the
	register value being different from the
	identity matrix.
	Note that $A_1=B^{-1}A^{-1}$, $A_2=A$ and $A_3=B$, where $A$ and $B$
	are the generators of the group $G_{BS}$ and recall that $G_{BS}$ is
	isomorphic to $BS(1,2)$. Since $ \mathtt{UPOW} $ is a unary nonregular
	language, it is not context-free and we conclude the result.
\end{proof}

Note that $\mathfrak{L}(\mathbb{Z}) \subsetneq \mathfrak{L}(BS(1,2)) $ 
since the subgroup generated by $ a $ in $ BS(1,2) $ is isomorphic to 
$ \mathbb{Z} $ and $\mathfrak{L} (BS(1,2)) $ contains a unary nonregular language.

\subsection{Automata on matrices of higher dimensions}

In \cite{MS97}, it is proven that $ \mathbf{F}_2 \times \mathbf{F}_2 $-automata are as powerful as Turing machines.

\begin{fact}\textup{\cite{MS97}}\label{fact: RE}		
	$ \mathfrak{L} ( \mathbf{F}_2 \times \mathbf{F}_2 )$ is the family of
	recursively enumerable languages.
\end{fact}

We make the following observation.

\begin{thrm}\label{thm: sl4z}
	$\mathsf{RE} =  \mathfrak{L} ( \mathbf{F}_2 \times \mathbf{F}_2 ) =
	\mathfrak{L} (SL(4,\mathbb{Z}) )$.
\end{thrm}

\begin{proof} The first equality is Fact \ref{fact: RE}.
	Recall from Section \ref{Section: 23matrices} that $ \varphi $ is an isomorphism from $ \mathbf{F}_2$ onto $
	\mathbf{G} $, the matrix group generated by the matrices $ M_a $ and $
	M_b $.
	Let $ \mathbf{G}' $ be the following group of matrices
	\[ \left \{
	\left [
	\begin{array}{clll}
	\multicolumn{2}{l}
	{\multirow{2}{*}{$M_1$}} & 0 & 0 \\
	& & 0 & 0 \\
	0&0 & \multicolumn{2}{c}{\multirow{2}{*}{ $M_2$}} \\
	0&0 & & \\
	\end{array}
	\right ], \ M_1, \ M_2 \in \mathbf{G} \right \}
	.\]
	
	We will define the mapping $ \psi:\mathbf{F}_2 \times \mathbf{F}_2
	\rightarrow \mathbf{G}' $ as $ \psi(g_1,g_2) =
	(\varphi(g_1),\varphi(g_2)) $ for all $ (g_1,g_2)\in \mathbf{F}_2
	\times \mathbf{F}_2 $ which is an isomorphism from $\mathbf{F}_2
	\times \mathbf{F}_2  $ onto $ \mathbf{G}' $.
	
	This proves that $\mathbf{F}_2 \times \mathbf{F}_2  $ is isomorphic to
	a subgroup of $ SL(4,\mathbb{Z}) $. The fact that $  \mathfrak{L} (
	\mathbf{F}_2 \times \mathbf{F}_2 ) $ is the set of recursively
	enumerable languages helps us to conclude that $
	\mathfrak{L}(SL(n,\mathbb{Z})) $  is the set of recursively enumerable
	languages for $ n \geq 4 $.
\end{proof}

Let us also state that the classes of languages recognized by automata
over supergroups of $ SL(4,\mathbb{Z}) $ such as $ GL(4,\mathbb{Z})
$ or $ SL(4,\mathbb{Q}) $ are also identical to the class of
recursively enumerable languages.
\begin{thrm}\label{thm: RE}
	$ \mathfrak{L}(G)  = \mathsf{RE} $, where $ G $ is any matrix group whose matrix entries are computable numbers and $ SL(4,\mathbb{Z})$ is a subgroup of $G $.
	\end{thrm}
	\begin{proof}
		Note that any finite automaton over a matrix group can be simulated by a nondeterministic Turing machine which keeps track of the register simply by multiplying the matrices and checking whether the identity matrix is reached at the end of the computation, provided that the matrix entries are computable numbers. Since $ \mathsf{RE}=\mathfrak{L} (SL(4,\mathbb{Z}) ) $ and $ G $ contains $ SL(4,\mathbb{Z})$ as a subgroup, $ \mathfrak{L}(G) $ is the set of recursively enumerable languages.
		\end{proof}

We summarize the results in Figure \ref{fig: diagram}. Solid arrows
represent proper inclusion, dashed arrows represent inclusion and
dashed lines represent incomparability. 
\begin{figure}[h!]
	\centering
	\includegraphics[width=1\linewidth]{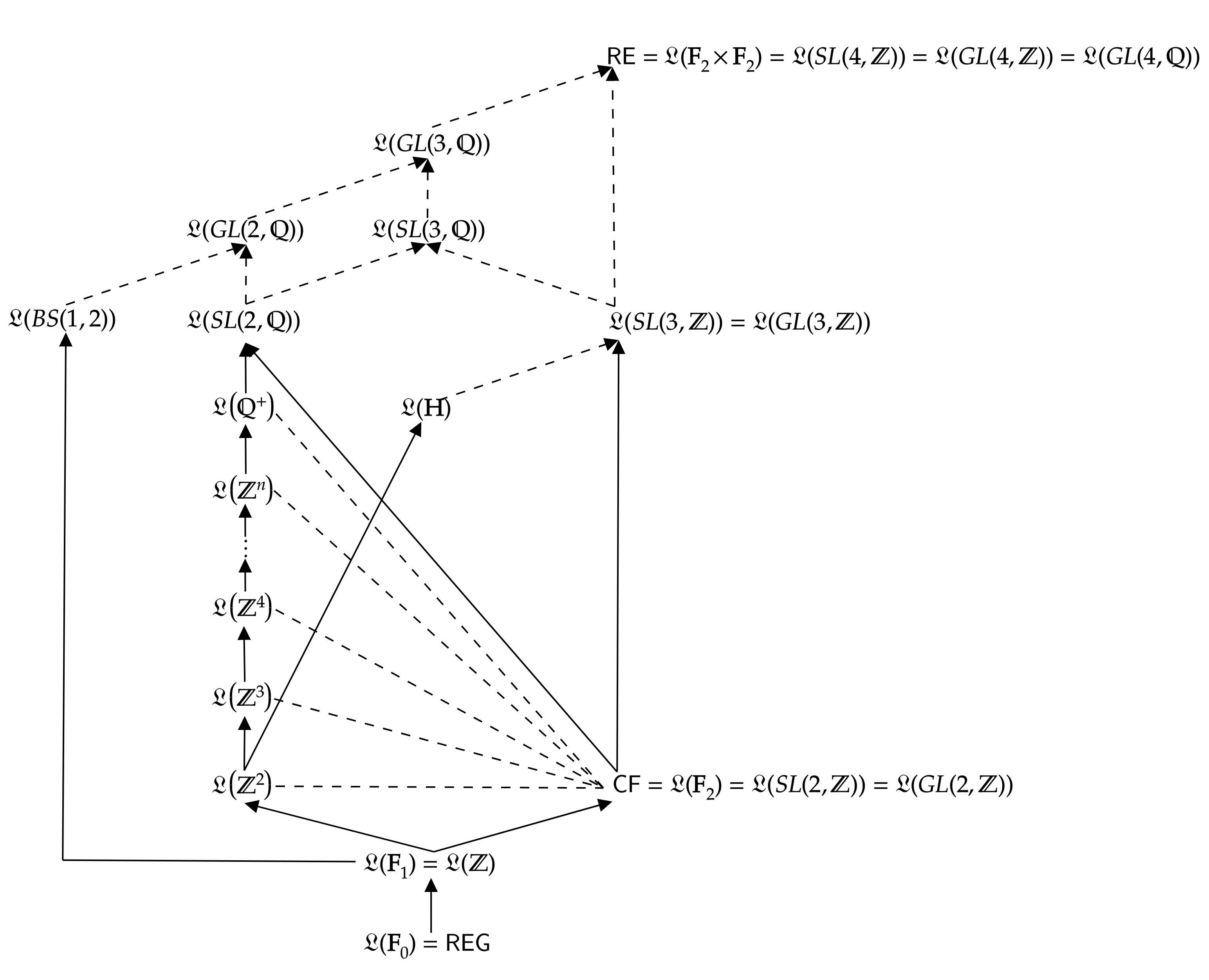}
	\caption{Language classes associated with groups}
	\label{fig: diagram}
\end{figure}

\section{Time complexity}
\label{Section: time}
In the previous section, we compared various automaton models solely on the basis of the groups they employed as a computational resource. The theory of computational complexity deals with various different types of such resources, the allowed runtime of the machines being the most prominent among them. Some of the automata we saw in Section \ref{Section: results} (e.g. Figure \ref{fig:upow}) have arbitrarily long computations, and it is a legitimate question to ask whether our results, for instance, the relationships in Figure \ref{fig: diagram}, would still hold if one imposed common time bounds on the automata. We study such questions in this section.

\subsection{Definitions}

A group automaton $ \mathcal{G} $ recognizing language $\mathtt{L} $ is said to be \textit{strongly $ t(n) $ time-bounded} if for any input string $ x $ with $ |x|=n $, every computation of $ \mathcal{G} $ on $ x $ takes at most $ t(n) $ steps. We will denote the set of languages recognized by strongly $ t(n) $-time bounded $ G $-automata by $ \mathfrak{L}(G)_{t(n)}^s $.

Although the strong mode of recognition defined above is standard in studies of time complexity, we will be able to prove the impossibility results of the next subsection even when the machines are subjected to the following, looser requirement:

A group automaton $ \mathcal{G} $ recognizing language $\mathtt{L} $ is said to be \textit{weakly  $ t(n) $ time-bounded} if for each accepted input string $ x \in \mathtt{L} $ with $ |x|=n $, $ \mathcal{G} $ has a successful computation which takes at most $ t(n) $ steps. So any input string is allowed to cause longer computations, as long as none of those are accepting for inputs which are not members of $\mathtt{L} $. We will denote the set of languages recognized by weakly $ t(n) $-time bounded $ G
$-automata by $ \mathfrak{L}(G)_{t(n)}^w $. 

A machine is \textit{real-time} if every transition consumes an input symbol. 

Note that the statement $ \mathfrak{L}(G)_{t(n)}^s \subseteq \mathfrak{L}(G)_{t(n)}^w  $ is true by definition.

Let $ X $ be a generator set for the group $ G $. The \textit{length} of $ g \in
G $, denoted $|g|_X$, is the length of the shortest representative for $
g $ in $ (X \cup X^{-1})^* $. Let $B^X_G(n)= \{g \in G,|g|_X\leq n \} $ be the set of all elements in 
$ G $ which can be represented by a word of length at most $ n $. The \textit{growth function of a group} $ G $ with
respect to a generating set $ X $, denoted $ g^{X}_G(n)$, is the cardinality 
of the set $ B^X_G(n) $, that is $ g^{X}_G(n) = | B^X_G(n)|$. The growth function is asymptotically independent of the generating set, and we will denote the growth function of a group $ G $ by $ g_G(n)$.

For a positive integer $ n $, two strings $ w,w' \in \Sigma^* $ are $
n $-\textit{dissimilar for $ \mathtt{L} $}  if $ |w|\leq n $, $ |w'|\leq n $, and there
exists a string $ v \in \Sigma^*  $ with $ |wv|\leq n
$, $ |w'v|\leq n $ such that $ wv \in \mathtt{L} $ iff $w'v \notin
\mathtt{L}  $. Let $ A_\mathtt{L}(n) $ be the maximum $ k $ such that
there exist $ k $ distinct strings that are pairwise $n $-dissimilar.

A finite set of strings $ S $ is said to be a set of \textit{uniformly $ n $-dissimilar strings for $ \mathtt{L} $} if for each string $ w\in S $, there exists a string $ v $ such that
$ |wv|\leq n$ and $wv \in \mathtt{L} $ and for any string $ w'\in S $ such that $w\neq w' $,  $ |w'v|\leq n$ and $ w'v \notin \mathtt{L} $. Let $ U_\mathtt{L}(n) $ be the maximum $ k $ such that there exist $ k $ distinct strings that are uniformly $n $-dissimilar.

Note that the following is always true by definition, since the strings in a uniformly $ n $-dissimilar set are pairwise $ n $-dissimilar.

\begin{lmm}\label{lemma: uniform}
	$ U_\mathtt{L}(n) \leq A_\mathtt{L}(n)$ for all $ n\geq 0 $. 
\end{lmm}

\subsection{Limitations of machines on slow groups running in short time}\label{section: mainthm}

\begin{thrm}\label{thm: growth2}
	Let $ G $ be a group with growth function $ g_G(n)$. $ \mathtt{L} \notin  \mathfrak{L}(G)_{t(n)}^w $ if $g_G(t(n)) \in o( U_\mathtt{L}(n)) $. 
\end{thrm}
\begin{proof}
	Suppose for a contradiction that there exists a weakly $ t(n) $ time-bounded $ G $-automaton $ \mathcal{G} $ recognizing $ \mathtt{L}  $ in time $ t(n) $. For a sufficiently large $ n $, let $S $ be the set of uniformly $ n $-dissimilar strings such that $ |S|=U_\mathtt{L}(n) $. For every string $ w_i \in S $, there exists a string $ v_i $ such that $ w_iv_i\in \mathtt{L} $ and $ w_jv_i\notin \mathtt{L} $ for all $ w_j \in S $ with $ i \neq j $ .
	
	Let $ S_{acc} $ be the set of accepted extended strings of the form $ w_iv_i \in \mathtt{L} $ with $ |w_iv_i|\leq n $ where $ w_i \in S $ and $ w_jv_i \notin \mathtt{L} $ for all $ w_j \in S $ with $ i\neq j $ and $ |w_jv_i|\leq n $. Let $ C $ be the set of $ t(n) $ time bounded accepting computation paths for the strings in $ S_{acc} $. The computation $ c_{w_iv_i } \in C $ on the string $ w_iv_i $  can be written as 
	\[ c_{w_iv_i} =c_{w_iv_i}^{w_i} c_{w_iv_i}^{v_i}\] where $ c_{w_iv_i}^{w_i} $ represents the computation up to the end of the prefix $ w_i $ and $ c_{w_iv_i}^{v_i} $ represents the rest of the computation on the string $ v_i $. 
	
	A configuration of a group automaton is a pair consisting of a state
	and a group element. Let us count the number of configurations that can be reached at the end of the computation $ c_{w_iv_i}^{w_i} $. Since the number of states is constant, the number of configurations that can be reached is dependent on the number of different group elements that can appear in
	the register. After reading a prefix $ w_i $ with $|w_i|= m\leq n $, the product of
	the labels on the edges can be given by $ l=g_{i_1}g_{i_2}\dots g_{i_{k}} $ for some $ k \leq t(m) $, since the computation in consideration is time bounded. $ l $ can be 
	expressed as a product of $ \kappa $ generators, where $ \kappa $ is at 
	most $C\cdot k $ for some constant $ C $, since each group
	element labeling a transition in $ \mathcal{G} $ is composed of at most some constant number of
	generators, which is independent of the length of the string. The number of 
	elements in $ G $ which can be represented as a product of at most $ \kappa$ 
	generators is given by $ g_G(\kappa) $ by the definition of the growth function 
	of $ G $. Hence, the number of different values that can appear in the register 
	after reading a string of length exactly $ m $ is less than or equal to 
	$ g_G(\kappa) $. Since  $ \kappa \leq  C\cdot k $ and $ k \leq t(m) $ and $g_G(t(n)) \in o( U_\mathtt{L}(n)) $, we can 
	conclude that $$ g_G(\kappa)  \leq  g_G(C \cdot t(m)) \in o(U_\mathtt{L}(n)). $$	
	
	Now it is easy to see that the number of different configurations that can be reached at the end of a computation $  c_{w_iv_i}^{w_i} $ is $  o(U_\mathtt{L}(n)) $. Note that the cardinality of the set $ S $, and thus that of $ S_{acc} $, is equal to $ U_L(n) $. Due to the pigeonhole principle, the same configuration must be reached at the end of two computations $  c_{w_{i}v_i}^{w_{i}} $ and $  c_{w_jv_j}^{w_j} $ for some $ i\neq j $. This will result in the acceptance of the strings $ w_{i}v_j $ and $ w_{j}v_{i} $, which are not members of $ \mathtt{L} $. We arrive at a contradiction and conclude that $\mathtt{L} $ cannot be recognized by any weakly $ t(n) $ time-bounded $ G $-automaton.
\end{proof}

In the next lemma, we set a lower bound on maximum cardinality of the set of uniformly $ n $-dissimilar strings in the word problem language of some group $ G $.

\begin{lmm}\label{lemma: growth2}
	Let $ G $ be a finitely generated group with growth function $ g_G(n)
	$. Then $ U_{W(G)}(n)\geq g_G(\lfloor \frac{n}{2}\rfloor)$.
\end{lmm}
\begin{proof}
	Let $ X $ be the generator set of $ G $. The number of distinct elements $ g $ in $ G $ which can be represented by a word of length less than or equal to
	$\lfloor \frac{n}{2}\rfloor $ is $ g_G(\lfloor \frac{n}{2}\rfloor)$, which is the cardinality of the set 
	$B_G^X(\lfloor \frac{n}{2}\rfloor)=\{g \in G,|g|_X\leq \lfloor \frac{n}{2}\rfloor\}$. Let $ S $ be the set containing the string representations of the elements in $ B_G^X(\lfloor \frac{n}{2}\rfloor) $. Every $ w_i \in S $ can be extended with $w_i^{-1}$ so that the extended string represents the identity element of $ G $ and has length less than or equal to $n $. 
	Since the strings in $ W(G) $ are those which belong to $ (X \cup X^{-1})^* $ and represent the identity element of $ G $, the extended string $ w_iw^{-1}_{i} \in W(G) $. For every string $ w_j\in S$ such that  $i \neq j$, $ w_jw_i^{-1} \notin W(G)$ since it is not possible for $ w_jw_i^{-1} $ to represent the identity element of $ G $. We conclude that the set $ S $ is uniformly $ n $-dissimilar. Since $ |S|=|B_G^X(\lfloor \frac{n}{2}\rfloor)| =  g_G(\lfloor \frac{n}{2}\rfloor)  $, it follows that $ U_{W(G)}(n)\geq g_G(\lfloor \frac{n}{2}\rfloor)$.
\end{proof}

The following theorem is about the language recognition power of finite automata over polynomial-growth groups which are weakly polynomial time-bounded.
\begin{thrm}
	Let $ G $ and $ H $ be groups with polynomial and exponential growth
	functions $ g_G(n)$ and $  g_H(n) $, respectively. For any polynomial $ t(n) $, $\mathfrak{L}(H)
	\nsubseteq \mathfrak{L}(G)_{t(n)}^w  $. 
\end{thrm}
\begin{proof}
	Since $ U_{W(H)}(n) \geq g_H(\lfloor \frac{n}{2}\rfloor) $ by Lemma $ \ref{lemma: growth2} $, and $ g_H(n) $ is an exponential function, $ U_{W(H)}(n)$ is also at least exponential.  $g_G(t(n))  $ is a polynomial function, since both $ g_G(n)$ and $ t(n) $  are polynomial. Hence, $ W(H) \notin \mathfrak{L}(G)_{t(n)}^w $ by Theorem \ref{thm: growth2}, and the result follows since $ W(H) $ is trivially in $ \mathfrak{L}(H) $.
\end{proof}

\begin{thrm}\label{thm: polycf}
	Let $ G $ be a group with a polynomial growth function. For any polynomial $ t(n) $,
	$ \mathsf{CF} \nsubseteq \mathfrak{L}(G)_{t(n)}^w $.
\end{thrm}
\begin{proof}
	It is known that the word problem of the free group of rank 2, $
	W(\mathbf{F}_2) $, has an exponential growth function \cite{Gi90}. Assuming that  $
	G $ is a group with polynomial growth function, $ W(\mathbf{F}_2) $
	cannot be recognized by any weakly $ t(n) $ time-bounded $ G $-automaton by
	Theorem \ref{thm: growth2}. Since $ W(\mathbf{F}_2) $ is a
	context-free language, the proof is complete.
\end{proof}

\subsection{Group automata under linear time bounds}

In this section, we focus on linear-time computation.

Let $ X $ be a generator set. For each symbol $ x\in X $, the functions $ P_x $ and $ Q_x $ are defined as follows.
\begin{align*}
P_x : X^*\rightarrow X^*  \hspace{0.5in}w\mapsto wx\\
Q_x : X^*x \rightarrow X^*\hspace{0.5in}wx\mapsto w
\end{align*}
Let $ K_X $ be the submonoid of all partial functions on $ X^* $ generated by $ P_x  $ and $ Q_x  $ for all $ x \in X $. $ K_X $ is called the \textit{polycyclic monoid} on $ X $. A $ K_X $-automaton is equivalent to a pushdown automaton, where $ P_x $ and $ Q_x $ can be interpreted as pushing and popping symbols on the stack. The equivalence between the two models is due to the nature of the functions $ P_x $ and $ Q_x $, and is described in detail in \cite{Ka09}. The resemblance between the free group and $ K_X $ is used to prove that $ \mathfrak{L}(\mathbf{F}_2)=\mathsf{CF} $ in \cite{Ka09} and \cite{Co05}.

Our aim is to show that $ \mathbf{F}_2 $-automata working in linear time can recognize all context-free languages. It is stated in \cite{Ze13} that $K_X$-automata which consume at least one input symbol at each step are as powerful as $K_X$-automata without any time bound. However, it is not straightforward to see whether the same is true for $ \mathbf{F}_2 $-automata.  

\begin{thrm}\label{thm: rtF2}
	$ \mathfrak{L}(\mathbf{F}_2)_{O(n)}^w =  \mathsf{CF}$.
\end{thrm}
\begin{proof}
	We are going to use the construction of Kambites \cite{Ka09} to prove that any context-free language can be recognized by a weakly linear-time bounded $ \mathbf{F}_2 $-automaton.  
	
	Let $ \mathtt{L} $ be a context-free language and let $ \mathcal{M}=\{Q,\Sigma, K_X, \delta,q_0,Q_a\} $ be a polycyclic monoid automaton recognizing $ \mathtt{L} $. $ K_X $ is the polycyclic monoid on $ X $ where the cardinality of the set $ X $ is $ n $ for some $ n \geq 2 $. Let $ e $ be the identity element of $ K_x $. The construction of Kambites provides an $ \mathbf{F}_{n+1}$-automaton $ \mathcal{G}=\{Q',\Sigma,\mathbf{F}_{n+1}, \delta',q_{0}',Q_a'\} $ recognizing the language $ \mathtt{L} $. The generator set for $ \mathbf{F}_{n+1} $ is $ X' $, where $ X'=X\cup \# $. 
	
	Let us analyze the construction in more detail.
	
	\begin{itemize}
		\item $ Q'=Q_{-} \cup Q_+ $ where $ Q_{-}=\{q_-|q \in Q\} $ and  $ Q_{+}=\{q_+|q \in Q\} $ 
		\item $ q_0' $=$ q_+ $ where $ q=q_0 $. 
		\item $ Q_a'=\{q_-|q\in Q_a \} $.
		\item$  \delta'(p_+,\sigma)=(q_+,x\#) $ if $ \delta(p,\sigma)=(q,x\#) $ where $ x $ is a positive generator for all $ \sigma \in \Sigma $.
		\item$ \delta'(p_-,\sigma)=(q_+,x'\#) $ if $ \delta(p,\sigma)=(q,x'\#) $ where $ x' $ is a negative generator for all $ \sigma \in \Sigma $.
		\item $ \delta'(p_+,\sigma)=(q_+,e) $ if $ \delta(p,\sigma)=(q,e) $ for all $ \sigma \in \Sigma $.
		\item $ \delta'(q_+,\epsilon)=(q_-,e) $ for each $ q\in Q $.
		\item $ \delta'(q_-,\epsilon)=(q_-,\#^{-1}) $ for each $ q\in Q $.
	\end{itemize}
	
	We  will prove that $ \mathcal{G} $ actually runs in linear time. There are two transitions where the automaton is allowed to move without consuming any input symbols.
	
	For each state $q\in Q $, there are two states $ q_+ $ and $ q_- $ in $ \mathcal{G}  $ which are connected with an edge labeled $ (\epsilon, e) $. These transitions do not change the register value, and cannot contribute more than half of the runtime of the machine, since at least one input symbol has to be consumed between any two executions of such transitions.
	
	$ \epsilon $-loops exist in the machine $ \mathcal{G} $ for each state $ q_- $ where the loop is labeled by $ (\epsilon, \#^{-1}) $. Although this looks worrisome at first for the purpose of bounding the runtime, the number of times these loops are traversed is actually bounded, as the following argument suggests. Suppose that the register is multiplied with $ l_1 $, $ l_2 $, $ \cdots $, $ l_m $ while reading some input string $ w$ of length $ n $, resulting in the register value $l=l_1l_2\cdots l_m(\#^{-1})^k$, where $ k \in \mathbb{N} $, at the end of the computation. If  $w $ is accepted by the machine, $ l $ should satisfy the following, as well as being equal to the identity element:
	\[
	l_i = \Biggl\{\begin{array}{lr}
	(\#^{-1})^px_i\# \mbox{ for some } p \in \mathbb{N}, & \mbox{if $ x_i $ is a negative generator}\\
	x_i\# , & \mbox{if $ x_i $ is a positive generator}\\
	\end{array}
	\]

	This is called a \textit{permissible padding} in \cite{Ka09}. By looking at the transition function of $ \mathcal{G} $, one can see that the register is multiplied by a $ \# $ only when an input symbol is consumed. Hence, the number of $ \# $'s that occur in $ l $ is less than or equal to the length of the string. The register is multiplied with $ \#^{-1} $ without consuming any input symbol. In order for the $ \# $'s and $ \#^{-1} $'s to cancel each other, they should be equal in number. Therefore, it can be concluded that the $ \epsilon $-loops are traversed at most $ n $ times.

	We can conclude that any context-free language can be recognized by a weakly linear-time bounded free group automaton. Since $ \mathbf{F}_2 $ contains every free group of countable rank, the proof is complete.
\end{proof}	

We state the following theorem, which is the linear-time equivalent of Fact \ref{fact: finite} \cite{Co05}.

\begin{thrm}\label{thm: rtindex}
	Suppose $ G $ is a finitely generated group and $ H $ is a subgroup of finite index. Then $ \mathfrak{L}(G)_{O(n)}^w = \mathfrak{L}(H)_{O(n)}^w	 $.
\end{thrm}
\begin{proof}
	We know that the statement is true in general when there is no time bound by \cite{Co05}. The proof in \cite{Co05} still works when all automata in the constructions are required to work in linear time.
\end{proof}

Now we can show that Theorem \ref{theorem:gl} also holds for linear-time bounded group automaton.

\begin{thrm}\label{thm: cfrt}
	$  \mathsf{CF}= \mathfrak{L}(\mathbf{F}_2)_{O(n)}^w= \mathfrak{L}(SL(2,\mathbb{Z}))_{O(n)}^w=\mathfrak{L}(GL(2,\mathbb{Z}))_{O(n)}^w  $.
\end{thrm}
\begin{proof}
	The proof is identical with the proof of Theorem \ref{theorem:gl} by using Theorem \ref{thm: rtindex}.
\end{proof}

By using the results proven in Subsection \ref{section: mainthm}, we can demonstrate the language recognition power of weakly linear-time bounded $ \mathbf{H} $-automata.

\begin{thrm}
	$ \mathfrak{L}(\textbf{H})_{O(n)}^w \subsetneq  \mathfrak{L}(SL(3,\mathbb{Z}))_{O(n)}^w $.
\end{thrm}		

\begin{proof}
	$ \mathfrak{L}(\mathbf{H})^w_{O(n)} \subseteq \mathfrak{L}(SL(3,\mathbb{Z}))_{O(n)}^w $ since $ \mathbf{H} $ is a subgroup of $ SL(3,\mathbb{Z}) $.  
	Since the Heisenberg group has polynomial growth function \cite{Ha00}, there exists a context-free language which can not be recognized by any $ \textbf{H} $-automaton in polynomial time by Theorem  \ref{thm: polycf}. 	Since $ \mathsf{CF}=\mathfrak{L}(SL(2,\mathbb{Z}))_{O(n)}^w$ by Theorem \ref{thm: cfrt}, the result follows. 
\end{proof}

\begin{thrm}
i.  For $ k\geq 5 $, $ \mathfrak{L}(\mathbf{H})^w_{O(n)} $ and $ \mathfrak{L}(\mathbb{Z}^k)^w_{O(n)}$  are incomparable.
\\
ii. $ \mathfrak{L}(\mathbf{H})^w_{O(n)} $  and $ \mathsf{CF} $ are incomparable.

\end{thrm}
\begin{proof}
\textit{ i.} 	In \cite{Re10}, a weakly linear-time bounded $ \mathbf{H} $-automaton which recognizes the language $\mathtt{MULT}=\{x^py^qz^{pq} | p,q\geq 0\}$ is constructed. The language $\mathtt{MULT}$ can not be recognized by any $ \mathbb{Z}^k $-automaton, since any bounded language in $ \mathfrak{L}(\mathbb{Q}^+) $ is semilinear by \cite{ISK76}. 
		
		In \cite{GS98}, it is implicitly proven there exists a uniformly $ n $-dissimilar set of size $ \Theta(n^k) $ for the language $ \mathtt{L}_k=\{0^{a_1}10^{a_2}1\dots 0^{a_k} 10^{a_1}10^{a_2}1\dots 0^{a_k} 1\}  $ for all integers $ k $. For $k=5 $, there exists a uniformly $ n $-dissimilar set of size $ \Theta(n^5) $ for the language $ \mathtt{L}_5$ and $ U_{\mathtt{L}_5}(n) \geq n^5 $. Since $ g_{\mathbf{H}}(n)$ is a polynomial of order 4 \cite{Ha00} and $ t(n)=O(n) $, $ g_{\mathbf{H}}(t(n)) \in o(U_{\mathtt{L}_5}(n) )$. By Theorem \ref{thm: growth2}, we conclude the result.
\\
		
\noindent $ ii. $ The language $\mathtt{MULT}=\{x^py^qz^{pq} | p,q\geq 0\}$ is not a context-free language. Since $ \mathbf{H} $ has a polynomial growth function \cite{Ha00}, there exists a context-free language which can not be recognized by any $ \mathbf{H} $-automaton in polynomial-time by Theorem \ref{thm: polycf}.

\end{proof}

Let us note that $ \mathtt{L}_5 $ can be recognized by a $ \mathbb{Z}^5 $-automaton in real time. The existence of the languages $ \mathtt{L}_k $ can be used to prove the linear-time nondeterministic counter hierarchy, with the help of Theorem \ref{thm: growth2}.

\begin{thrm}
	$ \mathfrak{L}(\mathbb{Z}^k )^w_{O(n)} \subsetneq \mathfrak{L}(\mathbb{Z}^{k+1} )^w_{O(n)} 	 $ for $ k\geq 1 $.
\end{thrm}
\begin{proof}
	The language $ \mathtt{L}_{k+1}=\{0^{a_1}10^{a_2}1\dots 0^{a_{k+1}} 10^{a_1}10^{a_2}1\dots 0^{a_{k+1}} 1\}  $	can be recognized by a $ \mathbb{Z}^{k+1}  $-automaton in real time. While scanning the first $ k+1 $ segments of $ 0 $'s, the $ i $'th counter is increased for each scanned $ 0 $ as $ 0^{a_i} $ is read. In the remainder of the computation, the $ i $'th counter is decreased for each scanned $ 0 $ when $ 0^{a_i} $ is read.
	
	There exists a uniformly $ n $-dissimilar set of size $ \Theta(n^{k+1}) $ for the language $ \mathtt{L}_{k+1} $, so $ U_{\mathtt{L}_{k+1}}(n) \geq n^{k+1} $. Since $ t(n)= O(n) $ and $ g_{\mathbb{Z}^{k}}(n) = n^k$ \cite{Gi90}, $ g_{\mathbb{Z}^{k}}(t(n)) \in o(U_{\mathtt{L}_5}(n) )$. We conclude by Theorem \ref{thm: growth2}.
\end{proof}

A celebrated result of the field of computational complexity, the nondeterministic time hierarchy theorem, will enable us to demonstrate that the computational power $ \mathbf{F}_2 \times \mathbf{F}_2$-automata is dependent on the time allotted for their execution. 

\begin{fact}\label{fact: nth} \textup{\cite{Zak83}}		
	If $g(n)$ is a time-constructible function, and $f(n+1) = o(g(n))$, then there exists a language which cannot be recognized by any nondeterministic Turing machine in time $f(n)$, but can be recognized by a nondeterministic Turing machine in time $g(n)$.
\end{fact}

Assume that any recursively enumerable language can be recognized by some linear-time $ \mathbf{F}_2 \times \mathbf{F}_2$-automaton. One can easily build a nondeterministic Turing machine that simulates such a  $ \mathbf{F}_2 \times \mathbf{F}_2$-automaton with only a polynomial slowdown. But this would mean that any recursively enumerable language can be recognized by  some nondeterministic TM in polynomial time, contradicting Fact \ref{fact: nth}, which implies that there exist languages which can only be recognized by nondeterministic Turing machines which run in at least exponential time. We have proven the following theorem.

\begin{thrm}
	$ 	\mathfrak{L}(\mathbf{F}_2 \times \mathbf{F}_2)_{O(n)}^w  \subsetneq \mathsf{RE} $.
\end{thrm}

Using the ability of Turing machines to simulate any finite automaton over a computable matrix group, the statement of the above theorem can be extended as follows.

\begin{thrm}
	$ \mathfrak{L}(G)_{O(n)}^w  \subsetneq \mathsf{RE} $ for any matrix group $ G $ whose matrix entries are computable numbers.
\end{thrm}
\begin{proof}
In Theorem \ref{thm: RE}, we have mentioned that Turing machines can simulate any finite automaton over a computable matrix group. By the nondeterministic time hierarchy theorem, it can be shown that there exist some languages which can not be recognized by any finite automata over matrix groups in linear time. 
\end{proof}

\begin{thrm}
	$ 	\mathfrak{L}(\mathbf{F}_2)_{O(n)}^w \subsetneq	\mathfrak{L}(\mathbf{F}_2 \times \mathbf{F}_2)_{O(n)}^w  $.
\end{thrm}
\begin{proof}
	It is obvious that an $ \mathbf{F}_2 $-automaton can be simulated by an $ \mathbf{F}_2 \times \mathbf{F}_2 $-automaton. $ \mathfrak{L}(\mathbf{F}_2)_{O(n)}^w = \mathsf{CF} $ by Theorem \ref{thm: cfrt}. The inclusion is proper since the non-context-free language $ \mathtt{L}=\{a^nb^nc^n|n\geq 0\} $ can be recognized by an $ \mathbf{F}_2 \times \mathbf{F}_2 $-automaton in real time by using the two registers as two counters. 
\end{proof}

In the rest of the section, the linear-time counterparts of the relationships in Figure \ref{fig: diagram} will be stated.

\begin{thrm}
	\begin{enumerate}[i.]
		\item $ \mathfrak{L}(\mathbb{Q}^+)^w_{O(n)} \subsetneq \mathfrak{L}(SL(2,\mathbb{Q}))^w_{O(n)} $.
		\item $ \mathfrak{L}(\mathbb{Z})^w_{O(n)} \subsetneq \mathfrak{L}(BS(1,2))^w_{O(n)} \nsubseteq \mathsf{CF}$.
		\item $ \mathfrak{L}(SL(2,\mathbb{Z}))^w_{O(n)} \subsetneq  \mathfrak{L}(SL(3,\mathbb{Z}))^w_{O(n)}$.
		\item $\mathfrak{L}(\mathbb{Z}^2)^w_{O(n)} \subsetneq \mathfrak{L}(\mathbf{H})^w_{O(n)}$.
		\item 	$\mathsf{CF}$ and $\mathfrak{L}(\mathbb{Z}^k)^w_{O(n)}$ are incomparable for all $ k \geq 2 $.
		\item $ \mathfrak{L}(SL(3,\mathbb{Z}))^w_{O(n)} = \mathfrak{L}(GL(3,\mathbb{Z}))^w_{O(n)}$.
		\item 	$ \mathsf{REG} = \mathfrak{L}(\mathbf{F}_0)^w_{O(n)} \subsetneq
		\mathfrak{L}(\mathbf{F}_1)^w_{O(n)} = \mathfrak{L}(\mathbb{Z})^w_{O(n)} \subsetneq
		\mathfrak{L}(\mathbf{F}_2)^w_{O(n)} $.
		
	\end{enumerate}
\end{thrm}
\begin{proof}
	
	$(i.,ii.,iii.,iv.) $ Analogous results where no time bound was imposed on the machines were proven in Theorems \ref{thm: qsl2q}, \ref{theorem: upow}, \ref{thm: sl2z3z}, and \ref{theorem: Z2H}, respectively. The group automata recognizing 
	the witness languages $\mathtt{L}=\{a^{2^{2n+1}} | n \geq 0 \}$,   $\mathtt{UPOW}=\{a^{2^n}|n\geq 0\}$ and $\mathtt{MULT}=\{x^py^qz^{pq} | p,q\geq 0\}$ operate in weakly linear time in all cases. 
	
	\textit{v.} The equivalent result for the general case is given in Fact \ref{fact: cfzn}. The non-context-free language $ \mathtt{L}'=\{a^nb^nc^n|n\geq 0\} $ can be recognized by a $ \mathbb{Z}^2 $-automaton in real time. 
	
	\textit{vi.} The equivalent result for the general case is given in Theorem \ref{thm: sl3zgl3z}. The result follows by Theorem \ref{thm: rtindex}.
	
	\textit{vii.} The equivalent result for the general case is given in Fact \ref{fact: reg}. $ \mathbf{F}_0 $ is the trivial group, and any regular language can be recognized by a deterministic finite automaton, which can be seen as finite automaton over $ \mathbf{F}_0 $, in real time. Since $ \mathbf{F}_1 $ is isomorophic to $ \mathbb{Z} $, the equality is obvious. Since the nonregular language $ \mathtt{L}=\{a^nb^n|n\geq 0\}$ can be recognized by a $ \mathbb{Z} $-automaton in real time, the proper inclusion follows. Lastly, since $ \mathfrak{L}(\mathbf{F}_2)^w_{O(n)}  $ is equivalent to $ \mathsf{CF} $ by Theorem \ref{thm: cfrt}, the last proper inclusion is still valid.
\end{proof}

The results are summarized in Figure \ref{fig:diagram-time}. 

\begin{figure}[h!]
	\centering
	\includegraphics[width=1\linewidth]{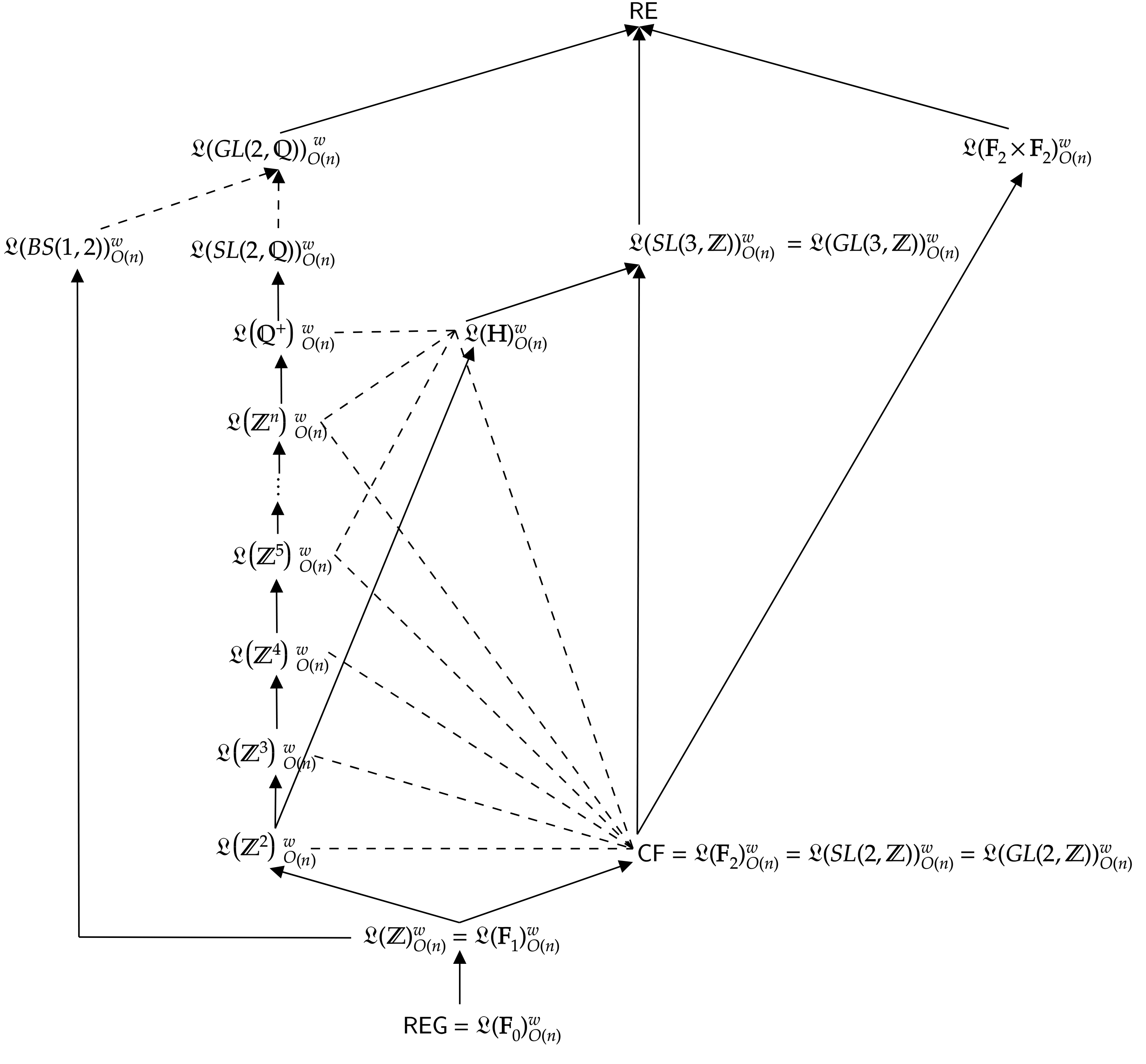}
	\caption{Language classes recognized by weakly linear-time bounded group automata }
	\label{fig:diagram-time}
\end{figure}

\section{Open questions}\label{Section: open}

Does there exist an $ SL(3,\mathbb{Z}) $-automaton recognizing $ W(\mathbb{Z}^3) $?
\footnote{Corollary 2 of \cite{CEO06} states that the word problem of a finitely generated 
	Abelian group $H$ is recognized by a $G$-automaton if and only if $H$ has a finite index 
	subgroup isomorphic to a subgroup of $G$. That corollary could be used to give an 
	affirmative answer to this open question. Unfortunately, the corollary is wrong: 
	Let $ H $ be an Abelian group and let $ G=\mathbf{F}_2 \times \mathbf{F}_2 $. 
	$ \mathfrak{L}(\mathbf{F}_2 \times \mathbf{F}_2 )$ contains the word problem of any 
	finitely generated Abelian group. Since $ \mathbf{F}_2 \times \mathbf{F}_2 $ is finitely 
	generated, any finite index subgroup of $ \mathbf{F}_2 \times \mathbf{F}_2 $ is also 
	finitely generated. Any finite index subgroup of $ \mathbf{F}_2 \times \mathbf{F}_2 $ is 
	either free or has a subgroup of finite index that is a direct product of free groups 
	\cite{BR84}. Any subgroup of an Abelian group is again Abelian. Hence, it is not 
	possible that $ G$ has a finite index subgroup isomorphic to a subgroup of $ H $.}

Can we prove a stronger version of Theorem \ref{thm: polycf}, which is independent 
of the time component? For instance, for the case of $ \mathbf{F}_2 $, is it true that 
$ W(\mathbf{F}_2) \notin \mathfrak{L}(\mathbf{H})$ in general?

Can we describe the necessary properties of a group $ G $ so that 
$ \mathfrak{L}(G) $ contains $ W(\textbf{F}_2) $?

Little is known about $ BS(1,2) $-automata. Does $  \mathfrak{L}(BS(1,2)) $ 
contain every context-free language?

Which, if any, of the subset relationships in Figure \ref{fig: diagram} are 
proper inclusions? 

Can we add other classes above RE in Figure 3 by examining groups on
matrices with uncomputable entries?

Theorem \ref{thm: growth2} uses the definition of uniform $ n $-dissimilarity requiring that $ g_G(t(n)) $  $\in o( U_\mathtt{L}(n))  $. Would the theorem be still true if we replace $ U_\mathtt{L}(n) $ by $ A_\mathtt{L}(n) $ ? The gap between  $ U_\mathtt{L}(n) $ and $ A_\mathtt{L}(n) $ might be large as mentioned in \cite{GS98}. Consider the language $ \mathtt{L}=\{a^ib^j|i\neq j\} $. It is stated in \cite{GS98} that a set of uniformly $ n $-dissimilar strings for $ \mathtt{L} $ can not contain more than two strings. However, $ A_\mathtt{L}(n) \notin O(1) $, since $ \mathtt{L} $ is not a regular language.

Can real-time $ \mathbf{F}_2 $-automata recognize every context-free language?

\bibliographystyle{plain}
\bibliography{references}

\end{document}